\documentclass[
final
]{dmtcs-episciences}

\received{2018-10-22}
\revised{2019-3-20}
\accepted{2019-3-25}
\publicationdetails{21}{2019}{1}{3}{4905}

\usepackage[normalem]{ulem}
\usepackage[utf8]{inputenc}
\usepackage{subfigure}
\usepackage{color}
\usepackage{tikz}
\usepackage{tkz-graph}
\newtheorem{theorem}{Theorem}
\newtheorem{corollary}[theorem]{Corollary}
\newtheorem{definition}[theorem]{Definition}
\newtheorem{lemma}[theorem]{Lemma}
\def\NN{\mathbb{N}}
\def\xx{+6}
\def\xxx{+12}

\def\yy{+4}

\def\And{\wedge}

\def\AND{\bigwedge}

\usepackage[round]{natbib}

\author{Rafael Ara\'ujo\affiliationmark{1}\thanks{rafaelteixeira@lia.ufc.br}
  \and Eurinardo Costa\affiliationmark{1}\thanks{eurinardo@ufc.br}
  \and Sulamita Klein\affiliationmark{2}\thanks{sula@cos.ufrj.br}
  \and Rudini Sampaio\affiliationmark{1}\thanks{rudini@dc.ufc.br}
  \and U\'everton S. Souza\affiliationmark{3}\thanks{ueverton@ic.uff.br}}

\title[FPT algorithms to recognize well covered graphs]{FPT algorithms to recognize well covered graphs}

\affiliation{
  Universidade Federal do Cear\'a, Fortaleza, Brazil\\
  Universidade Federal do Rio de Janeiro, Rio de Janeiro, Brazil\\
  Universidade Federal Fluminense, Niter\'oi, Brazil}
\begin{document}
\maketitle
\begin{abstract}
Given a graph $G$, let $vc(G)$ and $vc^+(G)$ be the sizes of a minimum vertex cover and a maximum minimal vertex cover of $G$, respectively. We say that $G$ is well covered if $vc(G)=vc^+(G)$ (that is, all minimal vertex covers have the same size). Determining if a graph is well covered is a coNP-complete problem.
In this paper, we obtain $O^*(2^{vc})$-time and $O^*(1.4656^{vc^+})$-time algorithms to decide well coveredness, improving results of 2015 by Boria et al. Moreover, using crown decomposition, we show that such problems admit kernels having linear number of vertices.
In 2018, Alves et. al. proved that recognizing well covered graphs is coW[2]-hard when $\alpha(G)=n-vc(G)$ is the parameter. Contrasting with such coW[2]-hardness, we present an FPT algorithm to decide well coveredness when $\alpha(G)$ and the degeneracy of the input graph $G$ are aggregate parameters.
Finally, we use the primeval decomposition technique to obtain a linear time algorithm for extended $P_4$-laden graphs and $(q,q-4)$-graphs, which is FPT parameterized by $q$, improving results of 2013 by Klein et al.
\end{abstract}

\section{Introduction}
Let $G=(V,E)$ be a graph. A subset $C$ is called a vertex cover if every edge of $G$ has an endpoint in $C$. A subset $I$ of $G$ is called an independent set if every pair of distinct vertices of $I$ are not adjacent in $G$. It is well known that $C$ is a vertex cover if and only if $V-C$ is an independent set.

Let $vc(G)$ be the size of a minimum vertex cover and let the independence number $\alpha(G)=n-vc(G)$ be the size of a largest independent set in $G$.

A vertex cover is minimal if it does not contain any distinct vertex cover of $G$. An independent set is maximal if it is not properly contained in any other independent set of $G$.
A graph $G$ is called well covered if all minimal vertex covers of $G$ have the same size $vc(G)$. That is, if $vc(G)=vc^+(G)$, where $vc^+(G)$ is the size of a maximum minimal vertex cover. Clearly, $vc(G)\leq vc^+(G)$ for every graph $G$.
Alternatively, a graph $G$ is well covered if all maximal independent sets of $G$ have the same size $\alpha(G)$.
The concept of well covered graph was introduced by \cite{plummer70}.

Well covered graphs are interesting because the greedy algorithm for producing a maximal independent set (resp. a minimal vertex cover) always produces a maximum independent set (resp. a minimum vertex cover). Recall that determining the independence number and the minimum vertex cover of a general graph are NP-hard problems.
Unfortunately, the problem of deciding if a graph is well covered is coNP-complete. This was independently proved by \cite{chvatal93} and by \cite{stewart92}. The problem remains coNP-complete even when the input graph is $K_{1,4}$-free (see \cite{caro96}).

Several papers investigate well coveredness in graph classes in order to obtain structural characterizations and polynomial time algorithms that recognize if a graph of such classes is well covered. See for example \cite{caro97,dean94,fradkin09,finbow09,prisner96,plummer93,randerath06,tankus97}.

We first consider some classes of graphs that have been characterized in terms of special properties of the unique primeval decomposition tree associated to each graph of the class. The primeval decomposition tree of any graph can be computed in time linear in the number of vertices and edges (see \cite{jamison95}) and therefore it is the natural framework for finding polynomial time algorithms of many problems. \cite{klein13} investigated the well coveredness of many classes of graphs with few $P_4$'s, such as cographs, $P_4$-reducible, $P_4$-sparse, extended $P_4$-reducible, extended $P_4$-sparse, $P_4$-extendible, $P_4$-lite and $P_4$-tidy.
In this paper, we extend results of 2013 by \cite{klein13} for two superclasses of those graph classes: extended $P_4$-laden graphs and $(q,q-4)$-graphs.
We obtain linear time algorithms to decide well coveredness for such graph classes. The algorithm for $(q,q-4)$-graphs is FPT parameterized by $q$.

We also obtain $O^*(2^{vc})$-time and $O^*(1.4656^{vc^+})$-time FPT algorithms to decide well coveredness, parameterized by $vc(G)$ and $vc^+(G)$, respectively, improving results of 2015 by \cite{boria15}.
Moreover, using crown decomposition, we show that such problems admit kernels having linear number of vertices. Contrasting with the coW[2]-hardness of recognizing well covered graphs by \cite{sula-cocoa16} when $\alpha(G)=n-vc(G)$ is the parameter, we obtain an FPT algorithm to decide well coveredness when $\alpha(G)$ and the degeneracy of the input graph $G$ are aggregate parameters, implying the fixed-parameter tractability, with respect to $\alpha(G)$, of graphs having bounded genus (such as planar graphs) and graphs with bounded maximum degree.

\section{The size of a minimum vertex cover as parameter}


We say that the running time of an FPT algorithm is $O^*(f(k))$, if it can be performed in $O(f(k)\cdot n^c)$-time, for some constant $c$.

In 2015, Boria et al. (see Theorem 5 of \cite{boria15}) proved that the maximum minimal vertex cover problem is FPT parameterized by $vc(G)$ (the vertex cover number). They obtained an $O^*(2.8284^{vc})$-time FPT algorithm to compute the maximum minimal vertex cover, which can be used to decide well coveredness of a graph. \cite{sula-cocoa16} investigated the well coveredness problem and proved that it is FPT parameterized by $vc(G)$, by obtaining an FPT algorithm with time $O^*(2^{nd})=O^*(2^{vc+2^{vc}})$, where $nd(G)$ is the neighborhood diversity of $G$.

In the following, we improve these results and show that it is possible to enumerate all minimal vertex covers in time $O(2^{vc}\cdot(m+n))$.

\begin{theorem}\label{teo-fpt}
It is possible to enumerate in time $O(2^{vc}\cdot(m+n))$ all minimal vertex covers of a graph. Consequently, there exists an $O^*(2^{vc})$-time FPT algorithm to decide  well coveredness parameterized by $vc=vc(G)$.
\end{theorem}

\begin{proof}
Let $C$ be a minimum vertex cover of $G$. Then all edges have an endpoint in $C$. Therefore, for every partition of $C$ in two sets $A$ and $B$ ($A\cup B=C$, $A\cap B=\emptyset$), $A\cup (N(B)\setminus B)$ is a vertex cover of $G$ if there are no edges with both endpoints in $B$.

Moreover, for every minimal vertex cover $C'$ of $G$, $A=C\cap C'$ and $B=C\setminus C'$ form a partition of $C$ such that $C'=A\cup(N(B)\setminus B)$, since $C'\setminus C\subseteq N(B)$ (because $C'$ is a vertex cover and is minimal).

Thus, we can enumerate all minimal vertex covers of $G$ by checking for every partition $(A,B)$ of $C$ if $A\cup (N(B)\setminus B)$ is a minimal vertex cover of $G$.
See Figure \ref{fig.particoes} for an example with a graph $G$, a minimum vertex cover $C$ and three different partitions $(A,B)$ of $C$. Only in the third partition, $A\cup(N(B)\setminus B)$ is a minimal vertex cover. In the first  partition, $A\cup(N(B)\setminus B)$ is not minimal and, in the second partition,  $A\cup(N(B)\setminus B)$ is not a vertex cover.
Notice that verifying if a set is a minimal vertex cover can be done in time $O(m+n)$. Since there are $2^{|C|}$ partitions of $C$, $|C|=vc(G)$ and it is possible to obtain a minimum vertex cover $C$ in time $O(2^{vc}\cdot(m+n))$, we are done.
\end{proof}

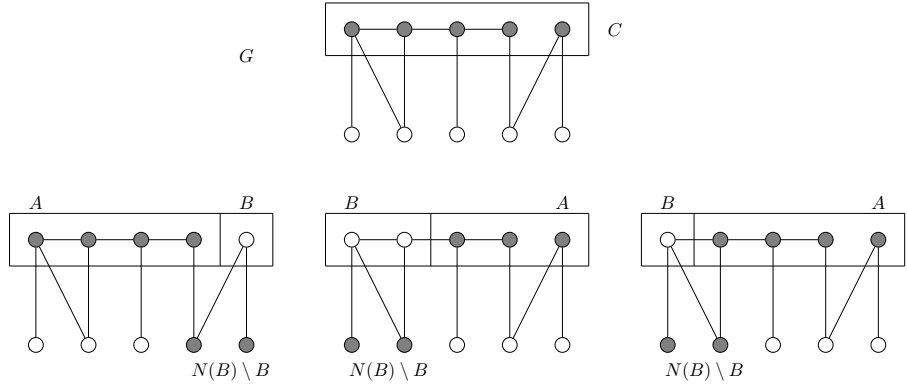
\begin{figure}\centering\scalebox{0.7}{
\begin{tikzpicture}[auto]
\tikzstyle{vertex}=[draw,circle,fill=white!25,minimum size=8pt,inner sep=1pt]
\tikzstyle{cvertex}=[draw,circle,fill=black!50,minimum size=8pt,inner sep=1pt]
\node[cvertex] (c1) at (2\xx,2\yy) {};
\node[cvertex] (c2) at (3\xx,2\yy) {};
\node[cvertex] (c3) at (4\xx,2\yy) {};
\node[cvertex] (c4) at (5\xx,2\yy) {};
\node[cvertex] (c5) at (6\xx,2\yy) {};
\node[vertex] (n1) at (2\xx,0\yy) {};
\node[vertex] (n2) at (3\xx,0\yy) {};
\node[vertex] (n3) at (4\xx,0\yy) {};
\node[vertex] (n4) at (5\xx,0\yy) {};
\node[vertex] (n5) at (6\xx,0\yy) {};

\path[-]
(c1) edge (c2) edge (n1) edge (n2)
(c2) edge (n2)
(c3) edge (c2) edge (c4) edge (n3)
(c4) edge (n4)
(c5) edge (n4) edge (n5);

\draw[-] (1.5\xx,1.5\yy) to (1.5\xx,2.5\yy);
\draw[-] (1.5\xx,2.5\yy) to (6.5\xx,2.5\yy);
\draw[-] (6.5\xx,2.5\yy) to (6.5\xx,1.5\yy);
\draw[-] (6.5\xx,1.5\yy) to (1.5\xx,1.5\yy);
 
\node () at (7\xx,2\yy) {$C$};
\node () at (0\xx,1.5\yy) {$G$};

\node[cvertex] (c1) at (2,2) {};
\node[cvertex] (c2) at (3,2) {};
\node[cvertex] (c3) at (4,2) {};
\node[cvertex] (c4) at (5,2) {};
\node[vertex]  (c5) at (6,2) {};
\node[vertex]  (n1) at (2,0) {};
\node[vertex]  (n2) at (3,0) {};
\node[vertex]  (n3) at (4,0) {};
\node[cvertex] (n4) at (5,0) {};
\node[cvertex] (n5) at (6,0) {};

\path[-]
(c1) edge (c2) edge (n1) edge (n2)
(c2) edge (n2)
(c3) edge (c2) edge (c4) edge (n3)
(c4) edge (n4)
(c5) edge (n4) edge (n5);

\draw[-] (1.5,1.5) to (1.5,2.5);
\draw[-] (1.5,2.5) to (6.5,2.5);
\draw[-] (6.5,2.5) to (6.5,1.5);
\draw[-] (6.5,1.5) to (1.5,1.5);
\draw[-] (5.5,1.5) to (5.5,2.5);
 
\node () at (2,2.7) {$A$};
\node () at (6,2.7) {$B$};
\node () at (5.7,-0.5) {$N(B)\setminus B$};

\node[vertex] (c1) at (2\xx,2) {};
\node[vertex] (c2) at (3\xx,2) {};
\node[cvertex] (c3) at (4\xx,2) {};
\node[cvertex] (c4) at (5\xx,2) {};
\node[cvertex] (c5) at (6\xx,2) {};
\node[cvertex] (n1) at (2\xx,0) {};
\node[cvertex] (n2) at (3\xx,0) {};
\node[vertex] (n3) at (4\xx,0) {};
\node[vertex] (n4) at (5\xx,0) {};
\node[vertex] (n5) at (6\xx,0) {};

\path[-]
(c1) edge (c2) edge (n1) edge (n2)
(c2) edge (n2)
(c3) edge (c2) edge (c4) edge (n3)
(c4) edge (n4)
(c5) edge (n4) edge (n5);

\draw[-] (1.5\xx,1.5) to (1.5\xx,2.5);
\draw[-] (1.5\xx,2.5) to (6.5\xx,2.5);
\draw[-] (6.5\xx,2.5) to (6.5\xx,1.5);
\draw[-] (6.5\xx,1.5) to (1.5\xx,1.5);
\draw[-] (3.5\xx,1.5) to (3.5\xx,2.5);
 
\node () at (2\xx,2.7) {$B$};
\node () at (6\xx,2.7) {$A$};
\node () at (2.7\xx,-0.5) {$N(B)\setminus B$};

\node[vertex]  (c1) at (2\xxx,2) {};
\node[cvertex] (c2) at (3\xxx,2) {};
\node[cvertex] (c3) at (4\xxx,2) {};
\node[cvertex] (c4) at (5\xxx,2) {};
\node[cvertex] (c5) at (6\xxx,2) {};
\node[cvertex] (n1) at (2\xxx,0) {};
\node[cvertex] (n2) at (3\xxx,0) {};
\node[vertex]  (n3) at (4\xxx,0) {};
\node[vertex]  (n4) at (5\xxx,0) {};
\node[vertex]  (n5) at (6\xxx,0) {};

\path[-]
(c1) edge (c2) edge (n1) edge (n2)
(c2) edge (n2)
(c3) edge (c2) edge (c4) edge (n3)
(c4) edge (n4)
(c5) edge (n4) edge (n5);

\draw[-] (1.5\xxx,1.5) to (1.5\xxx,2.5);
\draw[-] (1.5\xxx,2.5) to (6.5\xxx,2.5);
\draw[-] (6.5\xxx,2.5) to (6.5\xxx,1.5);
\draw[-] (6.5\xxx,1.5) to (1.5\xxx,1.5);
\draw[-] (2.5\xxx,1.5) to (2.5\xxx,2.5);
 
\node () at (2\xxx,2.7) {$B$};
\node () at (6\xxx,2.7) {$A$};
\node () at (2.7\xxx,-0.5) {$N(B)\setminus B$};

\end{tikzpicture}}
\caption{Graph $G$ with a minimum vertex cover $C$ and three different partitions $(A,B)$ of $C$. Only in the third partition, $A\cup(N(B)\setminus B)$ is a minimal vertex cover.}
\label{fig.particoes}
\end{figure}

\vspace{0.2cm}

\subsection{Polynomial kernel}

Crown decomposition is a general kernelization technique that can be used
to obtain kernels for many problems (see \cite{cygan2015parameterized}).

\begin{definition} (Crown decomposition) 
A crown decomposition of a graph $G$ is a parti\-tioning of $V(G)$ into three parts $C$ (Crown), $H$ (Head) and $R$ (Remainder), such that:
\begin{itemize}
\item $C$ is nonempty;
\item $C$ is an independent set;
\item There are no edges between vertices of $C$ and $R$. That is, $H$ separates $C$ and $R$;
\item Let $E$ be the set of edges between vertices of $C$ and $H$. Then $E$ contains a matching of size $|H|$. In other words, $G$ contains a matching of $H$ into $C$.
\end{itemize}
\end{definition}

The following lemma is the basis for kernelization using crown
decomposition (see \cite{cygan2015parameterized}).

\begin{lemma}~\label{crown lemma} (Crown lemma)
Let $k$ be a positive integer and let $G$ be a graph without isolated vertices and with at least $3k + 1$ vertices. There is a polynomial-time algorithm that either
\begin{itemize}
\item finds a matching of size $k + 1$ in $G$; or
\item finds a crown decomposition of $G$.
\end{itemize}
\end{lemma}

Now, by using Crown lemma, we will present a kernelization algorithm for recognizing well covered graphs, where the vertex cover number is the parameter.  
For simplicity, let $k$ be the size of a  minimum vertex cover of $G$, i.e.,  $k=vc(G)$.

\begin{corollary}\label{cor_wc}
Let $G$ be a graph without isolated vertices whose vertex cover number equals $k$, and with at least $3k + 1$ vertices. There is a polynomial-time algorithm that
finds a crown decomposition of $G$ such that $|H|\leq k$.
\end{corollary}

\begin{proof}
By Lemma~\ref{crown lemma} , there is a polynomial-time algorithm that either finds a matching of size $k + 1$ in $G$; or finds a crown decomposition of $G$.
As the size of a maximum matching of a graph is a lower bound to its vertex cover number, $G$ has no matching of size $k + 1$. 

By definition, in a crown decomposition of $G$ into $C$, $H$ and $R$, there is a matching of $H$ into $C$. As $C$ is an independent set, and there is no edge between $C$ and $R$, it is easy to see that there exists a minimum vertex cover of $G$ that contains $H$. Thus, $|H|\leq k$.
\end{proof}

\begin{lemma}\label{partitions_wc}
Let $G$ be a graph without isolated vertices and let $C$, $H$ and $R$ be a crown decomposition of $G$. If $G$ is well covered then $G[R]$ and $G[C\cup H]$ are well covered. 
\end{lemma}

\begin{proof}
Suppose that $G$ is well covered and $G[R]$ is not well covered. Then, $G[R]$ has two maximal independent sets $I_1$, $I_2$ such that $|I_1|\neq |I_2|$. As $C$ is an independent set that dominates $H$, it follows that $I_1\cup C$ and $I_2\cup C$ are two maximal independent set of $G$ having different cardinalities, contradicting the fact that $G$ is well covered.

Now, suppose that $G$ is well covered and $G[C\cup H]$ is not well covered. Then, $G[C\cup H]$ has a maximal independent set $I$ such that $|I|\neq |C|$. Clearly $I\setminus C\neq \emptyset$. Since $G[C\cup H]$ has a matching $M$ of $H$ into $C$, we know that each edge of $M$ has at most one vertex in $I$. Thus, $|I|< |C|$.
Let $R^*$ be the set of vertices in $R$ with no neighbors in $I$. Let $S_{R}$ and $S_{R^*}$ be a maximum independent set of $G[R]$ and $G[R^*]$, respectively. Note that $|S_{R^*}|\leq |S_R|$, then $S_R \cup C$ and $S_{R^*}\cup I$ are two maximal independent set of $G$ having different cardinalities, contradicting the fact that $G$ is well covered.
\end{proof}

\begin{lemma}\label{same_size}
Let $G$ be a graph without isolated vertices and let $C$, $H$ and $R$ be a crown decomposition of $G$ such that $R=\emptyset$. If $G$ is well covered then $|C|=$$|H|$.
\end{lemma}

\begin{proof}
It is well-known that there is a minimum vertex cover of $G$ that contains $H$ (see~\cite{cygan2015parameterized}). As $R=\emptyset$ then $H$ is a minimum vertex cover of $G$. By definition,  in a crown decomposition of $G$ into $C$, $H$ and $R$, there is a matching $M$ of $H$ into $C$. Suppose that there is an $M$-unsaturated vertex $w$$\notin V(M)$ ($w\in C$). The vertex $w$ has a neighbor in $H$, otherwise $w$ is an isolated vertex. Let $v$ be a neighbor of $w$ and let $K$ be a minimal vertex cover of $G$ such that $v\notin K$. Note that $N(v)\subseteq K$. For any vertex $x\in H$ (including $v$) such that $x\notin K$, there is a vertex $x^c \in C$ such that $(x,x^c) \in M$, which implies that $x^c\in K$. Thus, $K$ has size at least $|H|$. Since $w$ is $M$-unsaturated and it is also a vertex in $K$, then $K$ is a minimal vertex cover of $G$ of size greater than $|H|$, contradicting the fact that $G$ is well covered.
\end{proof}

\begin{corollary}\label{crown_wc}
Let $G$ be a well covered graph without isolated vertices whose vertex cover number equals $k$,  and with at least $3k + 1$ vertices. There is a polynomial-time algorithm that finds a crown decomposition of $G$ such that
\begin{itemize}
\item  $|C|=|H|\leq k$; and
\item $G[C\cup H]$ and $G[R]$ are well covered.
\item $G[R]$ has no isolated vertex.
\end{itemize}
\end{corollary}

\begin{proof}
By Corollary~\ref{cor_wc} there is a polynomial-time algorithm that finds a crown decomposition of $G$ such that $|H|\leq k$. By Lemma~\ref{partitions_wc} we know that $G[C\cup H]$ and $G[R]$ are well covered. Let $C'=C$, $H'=H$ and $R'=\emptyset$ be a crown decomposition of $G[C\cup H]$. By Lemma~\ref{same_size} follows that $|C|=|H|$. 

Now, it remains to show that $G[R]$ has no isolated vertex. 

Suppose that $G$ is well covered and $G[R]$ has an isolated vertex $\ell$. Any maximal independent set of $G[R]$ contains $\ell$. Let $S$ be a maximum independent set of $G[R]$. It is easy to see that $S\cup C$ is a maximum independent set of $G$. As $\ell$ is an isolated vertex in $G[R]$ and $G$ has no isolated vertex, then $\ell$ has a neighbor $u$ in $H$. Taking a maximal independent set $S_1$ of $G[C\cup H]$ that contains $u$, and a maximal independent $S_2$ of $G[R\setminus N(u)]$, it follows that $S_1\cup S_2$ is a maximal independent set of $G$. Note that $|S_1|=|C|$, because $G[C\cup H]$ is well covered, but $S_2$ has fewer vertices than $S$, since it does not contain $\ell$. Then, $G$ is not well covered.
\end{proof}

Now, we show a kernel having at most $5k$ vertices to the problem of determining whether $G$ is well covered.

\begin{theorem}\label{linearkernel}
Let $G$ be a graph without isolated vertices whose vertex cover number equals $k$. It holds that either
\begin{itemize}
\item $G$ is not well covered; or
\item $G$ has at most $5k$ vertices.
\end{itemize}
\end{theorem}
\begin{proof}
Suppose that $G$ is well covered and has at least $3k+1$ vertices. By Corollary~\ref{crown_wc} we can decompose $G$ into $C$, $H$ and $R$ such that $|C|+|H|\leq 2k$ and $G[R]$ is a well covered graph without isolated vertices. 

Let $C, H, R$ be a crown decomposition of $G$ that maximizes the size of the crown $C$. 

As $G[R]$ is a well covered graph without isolated vertices, then either $G[R]$ has at most $3k$ vertices; or $G[R]$ admits a crown decomposition into $C'$, $H'$ and $R'$.

If $G[R]$ has at most $3k$ vertices then $G$ has at most $5k$ vertices.

Otherwise, $G[R]$ admits a crown decomposition into $C'$, $H'$ and $R'$ and the following holds:
\begin{itemize}
\item $C \cup C'$ is nonempty;
\item $C \cup C'$ is an independent set, because $C'\subseteq R$.
\item There are no edges between vertices of $C \cup C'$ and $R'$;
\item $G$ contains a matching of $H\cup H'$ into $C \cup C'$.
\end{itemize}
Therefore, $C \cup C'$, $H\cup H'$ and $R'$ is also a crown decomposition of $G$. Since, by definition, $C'$ is nonempty then $C \cup C'$, $H\cup H'$ and $R'$ is a crown decomposition of $G$ with a larger crown. Thus, we have a contradiction.
\end{proof}

\section{The size of a maximum minimal vertex cover as parameter}

Following the results of \cite{boria15}, there is also an FPT algorithm to compute a maximum minimal vertex cover parameterized by $vc^+=vc^+(G)$ (the size of a maximum minimal vertex cover). Its time is $O^*(1.5397^{vc^+})$. In the following, we obtain a faster FPT algorithm with time $O^*(1.4656^{vc^+})$ to decide well-coveredness using the classical FPT algorithm for vertex cover (see \cite{cygan2015parameterized}).

\begin{theorem}\label{teo-fpt2}
Given a graph $G$ and $vc^+(G)$, it is possible to decide whether $G$ is well covered in time $O(1.4656^{vc+} \cdot n^2)$.
\end{theorem}

\begin{proof}
Let $G$ be a graph and let the parameter $k=vc^+(G)$.
The algorithm uses a search tree where each node $h$ has an associated graph $G_h$, a parameter $k_h$ and an associated ``\emph{partial vertex cover}'' $C_h$ in such a way that every edge of $G-G_h$ is covered by $C_h$. In the root $r$, the associated graph $G_r$ is equal to $G$, the parameter $k_r=k$ and $C_r$ is the empty set $\emptyset$. In the following, we define $G_h$ recursively for any non-root node $h$ of the search tree.
We say that a node $\ell$ is a leaf if and only if its associated graph $G_\ell$ has maximum degree at most 2 or its associated parameter $k_\ell=0$.

Let $h$ be a non-leaf node.
Then $G_h$ has a vertex $x_h$ with degree at least 3. Notice that there is no minimal vertex cover containing $N_{G_h}[x_h]$, since removing $x_h$ we also have a vertex cover. Then we have two possibilities: (1) $x_h$ in the vertex cover; or (2) $x_h$ not in the vertex cover and consequently $N_{G_h}(x_h)$ in the vertex cover, if $k_h\geq|N_{G_h}(x_h)|$.
We branch $h$ according to these two possibilities.

In the first child $h_1$, let $C_{h_1}=C_h\cup\{x_h\}$, $k_{h_1}=k_h-1$ and $G_{h_1}$ is obtained from $G_h$ by removing the vertex $x_h$. In words, we include $x_h$ to the partial vertex cover $C_{h_1}$ of $G_{h_1}$, and remove from $G_h$ the vertex $x_h$ to obtain $G_{h_1}$.

If $k_h\geq|N_{G_h}(x_h)|$, the node $h$ has a second child $h_2$ with $C_{h_2}=C_h\cup N_{G_h}(x_h)$, $k_{h_2}=k_h-|N_{G_h}(x_h)|$ and $G_{h_2}$ is obtained from $G_h$ by removing $N_{G_h}[x_h]$. In words, we include $N_{G_h}(x_h)$ to the partial vertex cover $C_{h_2}$ of $G_{h_2}$, and remove $N_{G_h}[x_h]$ from $G_h$ to obtain $G_{h_2}$.

For each leaf $\ell$ with $k_\ell=0$ and $G_\ell$ empty (no edges), we have that $C_\ell$ is a vertex cover of $G$. For each leaf $\ell$ with $k_\ell=0$ and $G_\ell$ non-empty, we have that $C_\ell$ has $k$ vertices and is not a vertex cover of $G$ (in this case, such a leaf is ignored).
For each leaf $\ell$ with $k_\ell>0$, $G_\ell$ has maximum degree at most 2 and then it is a simple graph consisting of paths, cycles or isolated vertices and it is possible to obtain a minimum and a maximum superset of $C_\ell$ which are minimal vertex covers of $G$ in linear time, if they exist. With this, we can obtain a minimum vertex cover (size $vc(G)$) and a maximum minimal vertex cover (size $vc^+(G)$) by looking the vertex covers obtained by the leaves which are minimal. With this, we can decide if $G$ is well covered.

Let $T(k_h)$ be the number of leaves in the subtree with root $h$.
Then $T(k_h)=T(k_{h_1})+T(k_{h_2})$ if $h$ has two children. Otherwise, $T(k_h)=T(k_{h_1})$. Moreover, for each leaf $\ell$, $T(k_\ell)=1$.
Since $|N(x_h)|\geq 3$, then $T(k_h)\leq T(k_h-1)+T(k_h-3)$.
Using induction on $k$, we have that $T(k)\leq 1.4656^{k}$.
Then the search tree has at most $1.4656^{k}$ leaves and consequently at most $2\cdot 1.4656^{k}+1$ nodes, since each node has at most two children. Since every non-leaf node takes time at most $O(n^2)$ and every leaf $\ell$ takes time $O(n^2)$ to decide if $C_\ell$ is a minimal vertex cover, we have that the total time is $O(1.4656^{k}\cdot n^2)$ and we are done.

Now consider a non-leaf node $h$ such that all vertices of $G_h$ have degree at most two.
If $G_h$ is connected, then $G_h$ is either a path or a cycle, and a minimum vertex cover and a maximum minimal vertex cover can be computed in $O(n^2)$ time. If $G_h$ is not connected, then it is a disjoint union of paths and cycles and we have the same time to compute both parameters.
\end{proof}

\begin{corollary}
Let $G$ be a graph without isolated vertices. The problem of determining whether $G$ is well covered admits a kernel having at most $5\cdot vc^+(G)$ vertices. 
\end{corollary}
\begin{proof}
It follows from Theorem~\ref{linearkernel} and the fact that $vc^+(G) \geq vc(G)$.
\end{proof}

\section{The size of a maximum independent set as parameter}

The local-treewidth (see \cite{eppstein00}) of a graph $G$ is the function $ltw_G:\NN\to\NN$ which associates with any $r\in\NN$ the maximum treewidth of an $r$-neighborhood in $G$. That is, $ltw_G(r)=\max_{v\in V(G)}\{tw(G[N_r(v)]\}$, where $N_r(v)$ is the set of vertices at distance at most $r$ from $v$. We say that a graph class $\mathcal{C}$ has bounded local-treewidth if there is a function $f_\mathcal{C}:\NN\to\NN$ such that, for all $G\in\mathcal{C}$ and $r\in\NN$, $ltw_G(r)\leq f_\mathcal{C}(r)$.
It is known that graphs with bounded genus or bounded maximum degree have bounded local-treewidth (see  \cite{eppstein00}). In particular, a graph with maximum degree $\Delta$ has $ltw_G(r)\leq\Delta^r$ and a planar graph has $ltw_G(r)\leq 3r-1$ (see \cite{bodlaender98}).

In this section, we consider bounded local-treewidth graphs and $d$-degenerate graphs, both classes include graphs with bounded genus and graphs with bounded maximum degree.

\begin{theorem}
Given a graph $G$ having bounded local-treewidth, the problem of determining whether $G$ is well covered is FPT when parameterized by $\alpha=\alpha(G)$. More precisely, it can be solved in $O(f(\alpha) \cdot n^2)$ time.
\end{theorem}

\begin{proof}
In the following, we express the well coveredness problem in First Order logic.
We use lower case variables $x,y,z,\ldots$ (resp. upper case variables $X,Y,Z,\ldots$) to denote vertices (resp. subsets of vertices) of a graph. The \emph{atomic formulas} are $x=y$, $x\in X$ and $E(x,y)$ which denotes the adjacency relation in a given graph. We say that a logic formula is FO (first order) if it is formed from atomic formulas with Boolean connectives $\wedge,\ \vee,\ \neg,\ \to,\ \leftrightarrow$ and \emph{element quantifications} $\exists x$ and $\forall x$. 

Consider the formula $Indep(X)$ which is true if and only if $X$ is an independent set:
\[
  Indep(X)\ :=\ \forall x,y\ (x\in X\And y\in X)\to\neg E(x,y)
\]
Also consider the formula $Maximal(X)$ which is true if $X$ is not properly contained in an independent set.
\[
  Maximal(X)\ :=\ \forall y\ \exists x\ (y\not\in X)\to (x\in X)\And E(x,y),
\]
where $(y\not\in X):=\neg(y\in X)$.

Given a graph $G$, it is not well covered if and only if $G$ has a maximal independent set $Y$ and a maximum independent set $X'$ with $|X'|\geq|Y|+1$. Thus $X'$ has a subset $X$ with $|X|=|Y|+1$, which is clearly independent.

With this, given a positive integer $k$, let $WellCov_{k}$ be the following first order formula, which is true if and only if the graph $G$ does not have two independent sets $X$ and $Y$ with $|X|=k$, $|Y|=k-1$ and $Y$ being maximal:
\[
WellCov_k\ :=\ \forall x_1,\ldots,x_k\ \forall y_1,\ldots,y_{k-1}\ 
\left(\AND_{1\leq i<j\leq k} x_i\ne x_j\right)\ \And\ Indep(\{x_1,\ldots,x_k\})
\]
\[
\to\ \neg\Big(Indep(\{y_1,\ldots,y_{k-1}\})\ \And\ Maximal(\{y_1,\ldots,y_{k-1}\}\Big)
\]

Notice that $WellCov_k$ contains $2k-1$ variables.

Now let $\alpha=\alpha(G)$. As mentioned before, if $G$ is not well covered, then there are independent sets $X$ and $Y$ with $2\leq |X|\leq\alpha$, $|Y|=|X|-1$ and $Y$ being maximal. With this, let $WellCov$ be the following first order formula, which is true if and only if $G$ is well covered:
\[
  WellCov\ :=\ \AND_{2\leq k\leq\alpha}\ WellCov_k.
\]

Then the well covered decision problem is first order expressible. Moreover, $WellCov$ contains at most $\alpha^2$ variables and then the size of the expression $WellCov$ is a function of $\alpha$.
We then can apply the Frick-Grohe Theorem (see Chapter 14 of \cite{downey-fellows13}) to prove that the well coveredness decision problem is FPT with parameter $\alpha(G)$ in time $O(n^2)$ for graphs with bounded local treewidth.
\end{proof}

\vspace{0.5cm}

The last theorem is a general result for bounded local treewidth graphs. We can obtain specific FPT algorithms (parameterized by $\alpha(G)$) for $d$-degenerate graphs, such as planar graphs, bounded genus graphs and bounded maximum degree graphs. 
A graph is called $d$-\emph{degenerate} if every induced subgraph has a vertex with degree at most $d$. The \emph{degeneracy} of a graph $G$ is the smallest $d$ such that $G$ is $d$-degenerate.
For example, outerplanar graphs, planar graphs and graphs with bounded maximum degree $\Delta$ have degeneracy at most $2$, $5$ and $\Delta$, respectively.



\begin{theorem}\label{teo-hered}
The problem of determining whether a given graph $G$ is well covered is FPT when parameterized by $\alpha=\alpha(G)$ and the degeneracy of $G$. More precisely, it can be solved in $O((d+1)^\alpha\cdot(m+n))$ time.
\end{theorem}

\begin{proof}
Let $G$ be a $d$-degenerate graph.
The algorithm uses a search tree $T$ with height $\alpha=\alpha(G)$ where each node $h$ of $T$ has an associated graph $G_h$. The associated graph of the root $r$ of $T$ is the original graph $G_r=G$. A leaf is a node such that its height is $\alpha$ or its associated graph is empty.

Let $h$ be a non-leaf node with associated graph $G_h$. We branch $h$ according to a vertex $v$ with minimum degree in the associated graph of $h$. Let $N_{G_h}[v]=\{u_1,\ldots,u_\ell\}$, where $\ell=|N_{G_h}[v]|\leq d+1$. With this, the node $h$ will have $\ell+1$ child nodes $h_1,h_2,\ldots,h_\ell$ in the search tree. In the child node $h_i$ ($1\leq i\leq \ell$), let $G_{h_i}=G_h$ initially and remove $N_{G_h}[u_i]$ from the associated graph $G_{h_i}$, which is also $d$-degenerate.

If there are two leaf nodes with different heights, return NO (since $G$ has a maximal independent set which is not maximum and then $G$ is not well covered). Otherwise, return YES.
Notice that the tree height is at most $\alpha(G)$ and each node has at most $d+1$ child nodes. Therefore, the search tree has at most $(d+1)^\alpha$ nodes and the total time is $O((d+1)^\alpha\cdot(m+n))$, since every node takes time $O(m+n)$.
\end{proof}

\vspace{0.5cm}

With this, we obtain the following corollary for graphs with bounded genus.

\begin{corollary}
Given a graph $G$ with bounded genus, the problem of determining whether $G$ is well covered can be solved in $O(7^\alpha\cdot(m+n))$ time.
\end{corollary}

\begin{proof}
Observe that if $G$ has genus equals $0$ then $G$ is a planar graph, thus $G$ is 5-degenerate and it is done by Theorem \ref{teo-hered}.
Now, assume that the genus $g$ of $G$ is bounded by a constant $c\geq 1$. Since $G$ has bounded genus $g$, we can assume that $n\geq 12g$, otherwise we can verify well coveredness in constant time.

Let $G$ be a graph embedded in a surface of genus $g$ without crossing edges (for example, toroidal graphs have $g=1$). From the Euler's formula for surfaces of genus $g$, we have that $m=n+f-2+2g$, where $f$ is the number of faces of $G$.  Moreover, we can assume that all faces of $G$ are triangles (otherwise we can increase the number of edges) and then every edge is part of two faces: $3f=2m$. Thus, $m=n+(2/3)m-2+2g$ and then $m=3n+6g-6$. As $n\geq 12g$, it follows that $m \leq (3.5)n-6$. Since the sum of vertex degrees is $2m\leq 7n-12$, then $G$ has a vertex of degree less than $7$. 
As bounded genus is a hereditary property, it holds that each subgraph of $G$ either has at most $12g$ vertices, or has a vertex of degree at most six. Therefore, by applying similar ideas in the proof of Theorem \ref{teo-hered}, we can determine whether $G$ is well covered in $O(7^\alpha\cdot(m+n))$ time.
\end{proof}


\section{Well coveredness of graphs with few $P_4$'s}

In this section, we obtain linear time algorithms for extended $P_4$-laden graphs and $(q,q-4)$-graphs.
A \emph{cograph} is a graph with no induced $P_4$ (see \cite{corneil81}).
A graph $G$ is \emph{$P_4$-sparse} if every set of five vertices in $G$ induces at most one $P_4$ (see \cite{olariu92}).
A graph $G$ is $(q,q-4)$ for some integer $q\geq 4$ if every subset with at most $q$ vertices induces at most $q-4$ $P_4$'s (see \cite{olariu01}).
Cographs and $P_4$-sparse graphs are exactly the $(4,0)$-graphs and the $(5,1)$-graphs.
\cite{olariu01} obtained polynomial time algorithms for several optimization problems in $(q,q-4)$-graphs.
A graph is extended $P_4$-laden if every induced subgraph with at most six vertices contains at most two induced $P_4$'s or is $\{2K_2, C_4\}$-free.
This graph class was introduced by \cite{giak96}.

A motivation to develop algorithms for extended $P_4$-laden graphs and $(q,q-4)$-graphs lies on the fact that they are on the top of a widely studied hierarchy of classes containing many graphs with few $P_4$'s, including cographs, $P_4$-sparse, $P_4$-lite, $P_4$-laden and $P_4$-tidy graphs. See Figure \ref{figura1}.
\cite{klein13} obtained linear time algorithms to determine well coveredness for $P_4$-tidy graphs.


\begin{figure}
\centering
\scalebox{1}{
\begin{tikzpicture}[auto]
\tikzstyle{set1}=[draw,rectangle,fill=black!00,minimum size=9pt,inner sep=5pt]
\tikzstyle{set2}=[draw,rectangle,fill=black!10,minimum size=9pt,inner sep=5pt]

\node[set1] (extladen) at (4,7.5) {Extended $P_4$-laden};
\node[set1] (qq4) at (9,7.5) {$(q,q-4)$-graph};

\node[set2] (tidy) at (2,6) {$P_4$-tidy};
\node[set1] (laden) at (6,6) {$P_4$-laden};
\node[set1] (q73) at (9,6) {$(7,3)$-graph};

\node[set2] (extendible) at (0,4.5) {$P_4$-extendible};
\node[set2] (extsparse) at (4,4.5) {Extended $P_4$-sparse};
\node[set2] (lite) at (8,4.5) {$P_4$-lite};

\node[set2] (extred) at (2,3) {Extended $P_4$-reducible};
\node[set2] (sparse) at (6,3) {$P_4$-sparse};

\node[set2] (red) at (4,1.5) {$P_4$-reducible};

\node[set2] (cograph) at (4,0) {Cograph};

\path[-]
(red) edge (cograph) edge (extred) edge (sparse)
(extsparse) edge (extred) edge (sparse) edge (tidy)
(tidy) edge (extendible) edge (extladen) edge (lite)
(laden) edge (lite) edge (extladen)
(q73) edge (lite) edge (qq4)
(lite) edge (sparse)
(extred) edge (extendible);


\node () at (10,6.8) {$q>7$};

\end{tikzpicture}
}
\caption{Hierarchy of graphs with few $P_4$'s. In gray, the classes investigated by \cite{klein13}.}
\label{figura1}
\end{figure}
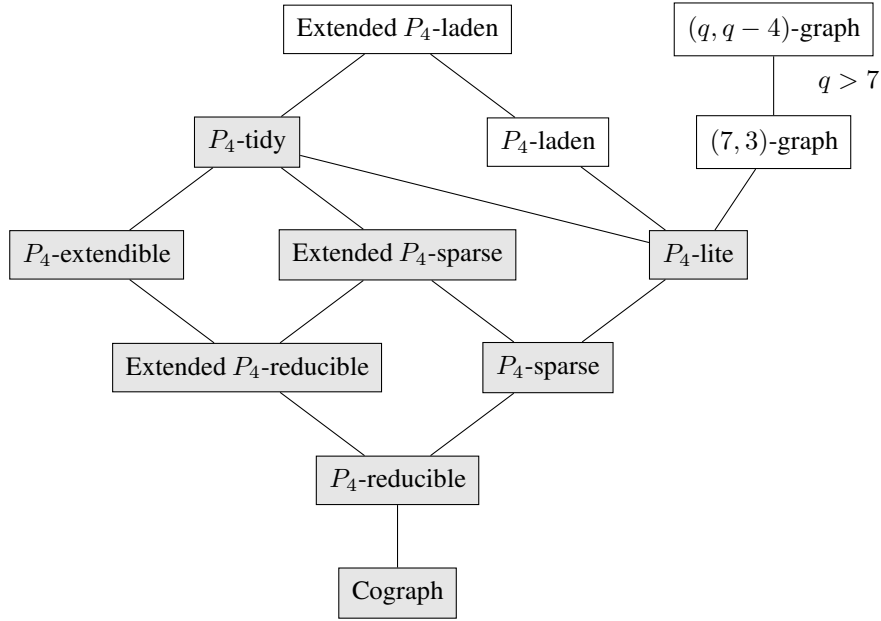


Given graphs $G_1=(V_1,E_1)$ and $G_2=(V_2,E_2)$, the \emph{union of $G_1$ and $G_2$} is the graph $G_1\cup G_2=(V_1\cup V_2,E_1\cup E_2)$ and the \emph{join of $G_1$ and $G_2$} is the graph $G_1\vee G_2=(V_1\cup V_2,E_1\cup E_2\cup\{uv: u\in V_1, v\in V_2\})$.

A \emph{pseudo-split} is a graph whose vertex set has a partition $(R,C,S)$ such that $C$ induces a clique, $S$ induces an independent set, every vertex of $R$ is adjacent to every vertex of $C$ and non-adjacent to every vertex of $S$, every vertex of $C$ has a neighbor in $S$ and every vertex of $S$ has a non-neighbor in $C$. Notice that the complement of a pseudo-split is also a pseudo-split.

A \emph{spider} is a pseudo-split with partition $(R,C,S)$ such that $C=\{c_1,\ldots,c_k\}$ and $S=\{s_1,\ldots,s_k\}$ for $k\geq 2$ and either $s_i$ is adjacent to $c_j$ if and only if $i=j$ (a thin spider), or $s_i$ is adjacent to $c_j$ if and only if $i\not=j$ (a thick spider).
Notice that the complement of a thin spider is a thick spider, and vice-versa.
A \emph{quasi-spider} is obtained from a spider by substituting a vertex of $C\cup S$ by a $K_2$ or a $\overline{K_2}$.

A graph $G$ is {\it $p$-connected} if, for every bipartition of the vertex set, there is a crossing $P_4$ (that is, an induced $P_4$ with vertices in both parts of the bipartition). A {\it p-component} of $G$ is a maximal $p$-connected subgraph.
A graph $H$ is {\it separable} if its vertex set can be partitioned in two graphs $(H_1,H_2)$ such that every induced $P_4$ $wxyz$ with vertices of $H_1$ and $H_2$ satisfies $x,y\in V(H_1)$ and $w,z\in V(H_2)$. We write $H\to(H_1,H_2)$.
It was proved that, if $G$ and $\overline{G}$ are connected and $G$ is not $p$-connected, then $G$ has a separable $p$-component $H\to(H_1,H_2)$ such that every vertex of $G-H$ is adjacent to every vertex of $H_1$ and non-adjacent to every vertex of $H_2$ (see \cite{jamison95}). 

\cite{giak96} proved that every $p$-connected extended $P_4$-laden graph is pseudo-split or quasi-spider $(R,C,S)$ with $R=\emptyset$, or is isomorphic to $C_5$, $P_5$ or $\overline{P_5}$. The following theorem suggests a natural decomposition for extended $P_4$-laden graphs, which can be obtained in linear time.

\begin{theorem}[\cite{giak96}]\label{teo-giak96}
A graph $G$ is extended $P_4$-laden if and only if exactly one of the following conditions is satisfied:
\begin{itemize}
\item $G$ is the disjoint union or the join of two extended $P_4$-laden graphs;
\item $G$ is pseudo-split or quasi-spider $(R,C,S)$ such that $G[R]$ is extended $P_4$-laden;
\item $G$ is isomorphic to $C_5$, $P_5$ or $\overline{P_5}$;
\item $G$ has at most one vertex.
\end{itemize}
\end{theorem}

For every $q\geq 4$, the $(q,q-4)$-graphs are self-complementary and have a nice structural decomposition in terms of union and join operations, spider graphs and separable p-components.
\cite{olariu98} proved that every $p$-connected $(q,q-4)$-graph is a spider $(R,C,S)$ with $R=\emptyset$ or has less than $q$ vertices.

\begin{theorem}[\cite{olariu98}]\label{teo-primeval}
Let $q\geq 4$.
If $G$ is a $(q,q-4)$-graph, then one of the following holds:
\begin{itemize}
\item $G=G_1\cup G_2$ is the union of two $(q,q-4)$-graphs $G_1$ and $G_2$;
\item $G=G_1\vee G_2$ is the join of two $(q,q-4)$-graphs $G_1$ and $G_2$;
\item $G$ is a spider $(R,C,S)$ such that $G[R]$ is a $(q,q-4)$-graph.
\item $G$ contains a separable $p$-component $H\to(H_1,H_2)$ with $|V(H)|<q$ such that $G-H$ is a $(q,q-4)$-graph and every vertex of $G-H$ is adjacent to every vertex of $H_1$ and not adjacent to any vertex of $H_2$;
\item $G$ has less than $q$ vertices.
\end{itemize}
\end{theorem}

This decomposition can be obtained in linear time (see \cite{olariu98}).
Linear time algorithms are obtained by \cite{olariu01} for several optimization problems in $(q,q-4)$-graphs using this decomposition.

\cite{klein13} characterized well coveredness for the union and join operations.

\begin{theorem}[Theorems 6 and 7 of \cite{klein13}]\label{lem-union-join}
Let $G_1$ and $G_2$ be two graphs.
Then $G_1\cup G_2$ is well covered if and only if $G_1$ and $G_2$ are well covered.
Moreover, $G_1\vee G_2$ is well covered if and only if $G_1$ and $G_2$ are well covered and $\alpha(G_1)=\alpha(G_2)$.
\end{theorem}

\vspace{0.5cm}

In the following, we characterize well coveredness for pseudo-split graphs.
It is worth mentioning that \cite{sula-cocoa16} obtained a characterization for well-covered split graphs which is similar to the following, but different, since pseudo-split graphs have the part $R$ which must be considered.

\begin{lemma}\label{lem-pseudosplit}
Let $G$ be a pseudo-split graph with partition $(R,C,S)$. Then $G$ is well covered if and only if $R=\emptyset$ and every vertex of $C$ has exactly one neighbor in $S$.
\end{lemma}

\begin{proof}
Recall that $C$ induces a clique, $S$ induces an independent set and every vertex of $C$ has a neighbor in $S$.
Suppose that $R\ne\emptyset$. Then, for every vertex $r\in R$, $S\cup\{r\}$ is an independent set of $G$ with $|S|+1$ vertices. Moreover, for any vertex $c\in C$, we have that $S\cup\{c\}\setminus N(c)$ is a maximal independent set of $G$ with at most $|S|$ vertices. Thus $G$ is not well covered.

Now assume that $R=\emptyset$. 
Let $I$ be a maximal independent set of $G$. Clearly $|I\cap C|\leq 1$ since $C$ induces a clique. Suppose that $I\cap C=\emptyset$. Then $I=S$ (because $I$ is maximal) and consequently $|I|=|S|$. Now suppose that $|I\cap C|=1$ and let $c\in I\cap C$. Then $S\setminus N(c)\subseteq I$ (because $I$ is maximal), and consequently $|I|$ is the number of non-neighbors of $c$ in $S$ plus one. Then $G$ is well covered if and only if $c$ has exactly one neighbor in $S$ for every $c\in C$.
Since they are the only possible maximal independent sets, we are done.
\end{proof}

In the following, we characterize well coveredness for quasi-spiders.

\begin{lemma}\label{lem-spider}
Let $G$ be a quasi spider with partition $(R,C,S)$. Then $G$ is well covered if and only if $R=\emptyset$ and $G$ is a thin spider with a vertex possibly substituted by a $K_2$.
\end{lemma}

\begin{proof}
Recall that $C$ induces a clique, $S$ induces an independent set and every vertex of $R$ is adjacent to each vertex of $C$ and non-adjacent to each vertex of $S$. Let $C=\{c_1,\ldots,c_k\}$ and $S=\{s_1,\ldots,s_k\}$ for $k\geq 2$. 

From Lemma \ref{lem-pseudosplit}, we have that a thin spider is well covered if and only if $R=\emptyset$. The same is valid if a vertex of $C\cup S$ is substituted by a $K_2$.
However, if a vertex of $C\cup S$ is substituted by a $\overline{K_2}$, then we obtain two independent sets with sizes $k$ and $k+1$, and consequently $G$ is not well covered.

From Lemma \ref{lem-pseudosplit}, we have that a thick spider is well covered if and only if $R=\emptyset$ and $k=2$, and consequently $G$ is also a thin spider  with $R=\emptyset$.
Moreover, the same is valid if any vertex of $C\cup S$ is substituted by a $K_2$ or a $\overline{K_2}$.
\end{proof}

\vspace{0.5cm}

In the following, we determine well coveredness for separable p-components $H$ with less than $q$ vertices.

\begin{lemma}\label{lem-pcomp}
Let $q\geq 4$ be a fixed integer and $G$ be a graph with a separable $p$-component $H$ with separation $H\to(H_1,H_2)$ with less than $q$ vertices such that $G-H\ne\emptyset$, every vertex of $G-H$ is adjacent to all vertices of $H_1$ and non-adjacent to all vertices of $H_2$. Then $G$ is well covered if and only if $G-H$ and $H_2$ are well covered and every maximal independent set of $H$ with a vertex of $H_1$ has exactly $\alpha(G-H)+\alpha(H_2)$ vertices.
\end{lemma}

\begin{proof}
At first, notice that, from any maximal independent sets $I$ of $G-H$ and $I_2$ of $H_2$, we obtain a maximal independent set $I\cup I_2$ of $G$.
Thus, if either $G-H$ or $H_2$ is not well covered, then $G$ is not well covered.
Moreover, any maximal independent set of $H$ with a vertex of $H_1$ is also a maximal independent set of $G$, since every vertex of $H_1$ is adjacent to all vertices in $G-H$. Thus, if there is a maximal independent set $I_1$ of $H$ containing a vertex of $H_1$ with $|I_1|\ne\alpha(G-H)+\alpha(H_2)$, then $G$ is not well covered.
Finally, notice that any maximal independent set of $G$ with a vertex in $G-H$ is obtained from a maximal independent set of $G-H$ and a maximal independent set of $H_2$, and we are done.
\end{proof}

\begin{theorem}
Let $G$ be a graph. If $G$ is a $(q,q-4)$-graph, we can determine well coveredness in linear time $O(2^qq^2\cdot(m+n))$.
If $G$ is extended $P_4$-laden, we can determine well coveredness in linear time.
\end{theorem}

\begin{proof}
Taking a primeval decomposition $\mathcal{T}$ of $G$, it is easy to see that one can construct a bottom-up dynamic programming according to the following rules:

\begin{enumerate}
\item if $G$ is isomorphic to $C_5$ or $\overline{P_5}$ then $G$ is well covered; 
\item if $G$ is isomorphic to $P_5$ then it is not well covered;
\item If $G=G_1\cup G_2$ is the union of two $(q,q-4)$-graphs or two extended $P_4$-laden graphs, then $G$ is well covered if and only if $G_1$ and $G_2$ are well covered; (see Theorem \ref{lem-union-join})
\item if $G=G_1\vee G_2$ is the join of two $(q,q-4)$-graphs or two extended $P_4$-laden graphs, then $G$ is well covered if and only if $G_1$ and $G_2$ are well covered and $\alpha(G_1)=\alpha(G_2)$; 
\item If $G$ is pseudo-split or quasi-spider, we are done by Lemmas \ref{lem-pseudosplit} and \ref{lem-spider};
\item if $G$ is a $(q,q-4)$-graph and has a separable p-connected component $H$, then we have from Lemma \ref{lem-pcomp} that well coveredness can be decided by verifying well coveredness for $H_2$ and checking all maximal independent sets of $H$ containing a vertex of $H_1$.
\end{enumerate}

Note that we can apply a linear-time preprocessing in order to check the well coveredness of every leaf node of $\mathcal{T}$. After that, in a bottom-up dynamic programming according to a post-order of $\mathcal{T}$, the well coveredness of nodes representing cases 1,2,3 and 4 can be checked in constant time. For a node representing case 5, to check if $R=\emptyset$ can be done in constant time, and if so, this node is a leaf. Finally, for nodes representing case 6, the well coveredness of $H_2$ can be tested in time $O(2^qq^2)$, since there are at most $2^q$ sets in $H_2$ and each one can be tested in time $O(q^2)$. We have the same time for maximal independent sets of $H$ containing a vertex of $H_1$.
\end{proof}


\acknowledgements
This research was supported by Faperj [CNE 09/2016] and [JCNE 03/2017], Funcap [4543945/2016] Pronem, CNPq Universal [425297/2016-0] and [437841/2018-9], and CAPES [88887.143992/2017-00] DAAD Probral.

\nocite{*}
\bibliographystyle{abbrvnat}

\end{document}